\documentclass[submission,copyright,creativecommons]{eptcs}
\providecommand{\event}{SCSS 2012} 
\usepackage{eurosym}
\usepackage{amssymb}
\usepackage{graphicx}
\usepackage{appendix}
\usepackage{multirow}
\usepackage{array}
\providecommand{\urlalt}[2]{\href{#1}{#2}}
\providecommand{\doi}[1]{doi:\urlalt{http://dx.doi.org/#1}{#1}}

\title{A New PVSS Scheme with a Simple Encryption Function}
\def\titlerunning{A New PVSS Scheme with a Simple Encryption Function}
\def\authorrunning{A. Ben Shil, K. Blibech, R. Robbana \& W.Neji}

\author{Assia Ben Shil
\institute{LIP2\\ Tunis, Tunisia}
\institute{Faculty of Sciences of Tunis\\
University of El Manar}
\email{essia.benshil@gmail.com}\\
\\
\and 
Kaouther Blibech  \qquad\qquad\qquad Riadh Robbana \qquad\qquad\qquad Wafa Neji
\institute{ \qquad\qquad LIP2\\   \qquad\qquad Tunis, Tunisia}
\email{kaouther.blibech@gmail.com \qquad riadh.robbana@fst.rnu.tn \qquad  neji.wafa@yahoo.fr}}

\newtheorem{theorem}{Theorem}[section]
\newtheorem{lemma}[theorem]{Lemma}

\newenvironment{proof}[1][Proof]{\begin{trivlist}
\item[\hskip \labelsep {\bfseries #1}]}{\end{trivlist}}
\newenvironment{definition}[1][Definition]{\begin{trivlist}
\item[\hskip \labelsep {\bfseries #1}]}{\end{trivlist}}

\begin{document}
\maketitle
\begin{abstract}
A Publicly Verifiable Secret Sharing (PVSS) scheme allows anyone to verify the validity of the shares computed
and distributed by a dealer. The idea of PVSS was introduced by Stadler in
[18] where he presented a PVSS scheme based on Discrete Logarithm. Later,
several PVSS schemes were proposed. In [2], Behnad and Eghlidos present an
interesting PVSS scheme with explicit membership and disputation processes.
In this paper, we present a new PVSS having the advantage of being simpler
while offering the same features.
\end{abstract}

\section{Introduction}

A secret sharing scheme is a cryptographic method allowing splitting a
secret between a set of participants such that only some predefined subsets
of participants can recover the shared secret. These qualified subsets are
called access structures. A secret sharing scheme proceeds in two phases: a
dealing phase in which a dealer computes shares and gives to every
participant his own share and a reconstruction phase that consists in trying
to reconstruct the shared secret by pooling the elements of a qualified
subset of shares.

Secret sharing schemes were introduced firstly and independently by Shamir
[16] and Blakley [3]. The first scheme is based on polynomial interpolation
while the latter is based on hyperplane geometry. Most of the proposed secret
sharing schemes [1, 11] are based on Shamir's secret sharing scheme. Although
its efficiency, Shamir's scheme still presents some problems. In fact, there
is an absolute trust in the dealer. This latter, can distribute some
inconsistent shares leading the participants to recover a secret which
differs from the initial one. Verifiable Secret Sharing (VSS) schemes [5, 6,
13] were proposed to allow participants to verify the validity of the shares
they received from the dealer. However, a malicious shareholder can receive
a valid share but submit an invalid one in the reconstruction phase.
Publicly Verifiable Secret Sharing (PVSS) schemes [2, 4], [8-10], [14, 15], [17-21 ] were
proposed to solve this problem. In fact, PVSS schemes were proposed to
prevent cheating by the dealer or/and the shareholders. In a PVSS scheme,
the validity of the distributed shares can be verified by anyone.

In [2], Behnad and Eghlidos present an interesting PVSS scheme where
participants can prove their membership and the validity of their shares to
prevent unauthorized parties from participating in the reconstruction
process. Moreover, their scheme offers an explicit disputation process
aiming to prove to a third party in conflict situations between the dealer
and a participant who among them is lying.

In this paper, we present a new PVSS scheme providing a disputation and a
membership proof processes. We show that our PVSS scheme is simpler than the
PVSS scheme presented in [2] while still being as secure as the mentioned
scheme.

This paper is organized as follows: First, PVSS schemes are presented. After
that, our new PVSS scheme is introduced. Then, the security of our PVSS
scheme is studied and a comparison between it and the previous PVSS schemes
is done. Finally, we provide some concluding remarks.

\section{PVSS Schemes}

PVSS schemes as introduced by Stadler in [18] aim to allow anyone, not only
participants, to verify that shares were correctly distributed by the
dealer. This property has been defined by Stadler in [18] and has been
denoted public verifiability.

Stadler proposed in this paper, two PVSS schemes that can be used with
general access structures. The first one is used for sharing a discrete
logarithm. It requires a non standard assumption called DDLP
\textquotedblleft Double Discrete Logarithm Assumption\textquotedblright .
In fact, Stadler dealt with expressions of the form $y=g^{(h^{x})}$ (with $g$
a generator of a group of order $p$, and $h$ a fixed element of high order
in $Z_{p}^{\ast }$) such that given $y$, it is hard to find $x$. Under this
assumption, his scheme is as secure as the Decisional-Diffie-Hellman
problem. The second one is based on the RSA root problem. It is used for
sharing the $n$-th root and depends on the RSA assumption. Encryptions are
based on a variant of the Diffie-Hellman key-exchange protocol. But we
should notice here that the security of this scheme was not formally
studied. Moreover, the verification in these two schemes requires
information exchanges between the verifier and the shareholder. We say that
it is an interactive verification.

In [8], Fujisaki and Okamoto defined the non-interactivity for a PVSS scheme
as the fact that the verification of a share can be done without
communicating with the dealer or with any other participant. The scheme they
proposed in [8] depends on the \textquotedblleft modified RSA
assumption\textquotedblright\ assuming that inverting the RSA function is
still hard. This modified RSA assumption allows partial recovery.

Notice that the schemes of [8, 18] depend on some non standard assumptions.
However, Schoenmakers provided in [15] a stronger PVSS scheme by adding the
fact that when submitting his share, the shareholder must provide its
correctness proof. His PVSS scheme is simpler than the previous schemes. It
uses techniques working in any group for which the Discrete Logarithm
Problem is hard. This scheme is as hard to break as the Decisional
Diffie-Hellman problem.

In [20], Young and Yung proposed an improvement of Schoenmakers's PVSS scheme.
The scheme they proposed to share discrete logarithm is as hard as the
Discrete Logarithm Problem itself. They proved in [21] that their scheme is
computational zero-knowledge. In addition, in PVSS schemes, secure
encryption assumptions are employed. But in their scheme, Young and Yung can
use any probabilistic encryption function.

In [4], Boudot and Traor\'{e} proposed new PVSS schemes allowing shareholders
to recover their shares quickly (fast recovery) or after a predetermined
amount of computations (delayed recovery). In fact, they provide a PVSS
scheme for sharing discrete logarithm with fast recovery and a PVSS scheme
for sharing factorization with fast recovery. They also present a PVSS
scheme for sharing discrete logarithm with delayed recovery and a PVSS
scheme for sharing factorization with delayed recovery.

In most of the existing PVSS schemes, the verification phase is interactive.
This is due to the use of Fiat-Shamir zero knowledge protocol [7]. In [14],
Ruiz and Villar proposed a PVSS scheme with non interactive verification. It
is the first efficient PVSS that does not use the Fiat-Shamir technique. It
is based on the homomorphic properties of Paillier's encryption scheme [12].
It is the first known PVSS scheme based on the DCRA\footnote{%
The Decisional Composite Residuosity Assumption, used in the proof of the
Paillier cryptosystem, says that given an integer $z$ and a composite $n$,
it is hard to decide whether $z$ is a n-residue modulo n$^{2}$ or not.}
(Decisional Composite Residuosity Assumption). The verification process in
this scheme is simpler than in the other known schemes.

In [9] , Heidarvand and Villar proposed a new PVSS scheme based on pairing.
They took back the scheme of Shoenmakers using the pairing. The security of
this scheme is based on the DBSDH \footnote{%
Let $e:G_{1}$ $\ast $ $G_{1}\rightarrow G_{2}$ a bilinear application such
that $G_{1}$ and $G_{2}$ are two multiplicative group with the same order $p$%
. Let $g$ be a generator of $G_{1}$ and $a$, $b$ and $z$ elements of $%
Z_{p}^{\ast }$. The Decisional Bilinear Square Diffie-Hellman (DBSDH)
problem says that $g^{a}$, $g^{b}$ and $e(g,g)^{z}$ is hard to decide
whether $e(g,g)^{a^{2}b}=e(g,g)^{z}$.} problem (Decisional Bilinear Square
Diffie-Hellman problem). In [10], Jhanwar proposed a new non-interactive
PVSS scheme based on pairing. In this scheme, the dealer has not to compute
and to distribute the shares of a given secret; he provides a set of private
keys for participants. Then, every participant uses his private key, joined
to another public value to compute his share.

Recently, other PVSS schemes have been proposed. In [21], Yu and all proposed
a publicly verifiable secret sharing scheme with the possibility of
enrollment. In [19], Wu and all proposed a pairing based PVSS scheme reducing
the computation cost while keeping the same security level of the existing
public key systems.

Behnad and Eghlidos provided, in [2], a PVSS scheme with non interactive
verification and having two peculiarities. First, after distributing the
shares and in case of any complaint from any participant, a third party can
run a disputation process to identify who is lying. This third party can
then vote against the dealer or against the participant. Second, Behnad and
Eghlidos added a membership proof process in the beginning of the
reconstruction phase. In this phase a shareholder has to prove his
membership and the validity of his share at the same time. In [17], Ben
Shil, Blibech and Robbana proposed another PVSS scheme with a disputation
and a membership proof processes. In this scheme, rather than publishing the
encrypted coefficients of the polynomial used to compute the shares, the
encrypted shares are published. Thus, the set of shares is public and any
insertion or deletion will be detected by all the old participants. This
scheme is, then, recommended for applications where the number of
participants is limited while the access structure is dynamic and where it
is worthy to keep a track of any change in the set of participants.

In this paper we introduce a new PVSS scheme providing a non-interactive
verification process and presenting explicit disputation and membership
processes. We show that our PVSS scheme is simpler than the schemes proposed
in [2] and [17] while keeping the same level of security.

\section{A new PVSS scheme}

In our scheme, given two large prime numbers $p$ and $q$ such that $q|p-1$%
\footnote{%
q divides p-1.}, the following notations are used:

\begin{description}
\item[-] $G_{q}$ is a subgroup of prime order $q$ in $Z_{p}^{\ast }$, such
that computing discrete logarithm in this group is infeasible and $g$ $\in $ 
$G_{q}$ is a generator of the group.
\end{description}

In our PVSS scheme, we perform all the computations in $Z_{q}$.

\subsection{Dealing phase}

\subsubsection{Distribution process}

In the distribution process, the dealer sets $%
F(x)=F_{0}+F_{1}x+...+F_{k-1}x^{k-1}$, where $F_{1}$, \ldots , $F_{k-1}\in
_{R}\footnote{%
Randomly chosen.}Zq$ and $F_{0}$ is the secret to share. Moreover:

\begin{enumerate}
\item Every participant chooses a private key $a_{i}$ where $a_{i}\in _{R}Zq$
and publishes $g^{a_{i}}$ as his public key, for $1\leq i\leq n$ where $n$ is the number of participants.

\item The dealer $D$ computes the shares $s_{i}$ $=F(i)$, for $0\leq i\leq n$%
.

\item He publishes $C_{j}=g^{F_{j}}$, for $0\leq j\leq k-1$ and $g^{s_{i}}$, for $0\leq i\leq n$.

\item He sends an encrypted  
share $E_{i}=s_{i}\oplus (g^{a_{i}})^{s_{i}}$ to
the participant $Pr_{i}$, for $1\leq i\leq n$ (Notice that ${s_{0}}$ is the secret and thus there is no associated encrypted share to be sent to anyone).
\end{enumerate}

\subsubsection{Verification process}

Every shareholder $Pr_{i}$, computes $s_{i}=E_{i}\oplus \lbrack
(g^{s_{i}})^{a_{i}}]$, then, verifies the following equality\footnote{%
Given $C_{j}$ $=$ $g^{F_{j}}$, we compute:
\par
$\prod_{j=0}^{k-1}(C_{j})^{i^{^{_{j}}}}=$ $\prod_{j=0}^{k-1}(g^{F_{j}})^{i^{j}}=$ $\prod_{j=0}^{k-1}g^{F_{j}}\ast {i^{j}}=g$ $%
^{\sum_{j=0}^{k-1}F_{j}\ast i^{j}}=$ $g^{F(i)}=$ $g^{s_{i}}$,(since $%
s_{i}=F(i)$).}: $g^{s_{i}}=\prod_{j=0}^{k-1}(C_{j})^{i^{j}}$.
Otherwise, the shareholder complains against the dealer.

\subsubsection{Disputation process}

In the case of any complaint, both the dealer $D$\ and the shareholder $%
Pr_{i}$\ try to prove their honesty to a third party $R$. For doing that, $D$%
\ has to publish an encrypted value leading $Pr_{i}$\ to extract $g^{s_{i}}$%
\ and to verify the validity of the associated share $s_{i}$. If $D$\ sends
an invalid share, $Pr_{i}$\ has to prove this fact to $R$. This process is
done using the following protocol:

\begin{enumerate}
\item $Pr_{i}$ chooses his private key $a_{i}$ and publishes his public key $%
g^{a_{i}}$.

\item $Pr_{i}$ and $D$ publish independently\ $%
g^{[(g^{a_{i}})^{s_{i}}]^{-1}}. $ Then, $R$ verifies that $D$ and $Pr_{i}$
published the same value. Else, $Pr_{i}$ sends $a_{i}$ to $R$.  $R$ computes $g^{a_{i}}$ and $g^{[(g^{a_{i}})^{s_{i}}]^{-1}}$ in order to discover who is lying. Notice that $R$ can compute $g^{s_{i}}$ from the
published values $g^{s_{i}}$ $=\prod_{j=0}^{k-1}(C_{j})^{i^{^{_{j}}}}$.\item $D$ computes and publishes $\lambda =s_{i}\oplus (g^{a_{i}})^{s_{i}}$%
.\ \ \ \ 

\item $Pr_{i}$ computes $\alpha $ $=$\ $\lambda \oplus (g^{a_{i}})^{s_{i}}$.
If $g^{\alpha }$ $=$ $\prod_{j=0}^{k-1}(C_{j})^{i^{^{_{j}}}}$, he sends a
commitment to $R$ and the disputation process is stopped. Else, he sends $%
\alpha $ to $R$.

\item $R$ computes $g^{\alpha }$ and verifies that $g^{\alpha }\neq
\prod_{j=0}^{k-1}(C_{j})^{i^{^{_{j}}}}$. Then, he verifies that $%
g^{1/(\lambda \oplus \alpha )}=g^{[(g^{a_{i}})^{s_{i}}]^{-1}}$. If it holds, 
$D$ lied else $Pr_{i}$ lied.
\end{enumerate}

\subsection{Reconstruction phase}

\subsubsection{Membership only proof}

If a verifier wants to verify that $Pr_{i}$ is an authorized participant,
this latter has to prove his membership to the verifier without revealing
his share. Our membership proof is the following:

\begin{enumerate}
\item The verifier chooses $a\in _{R}Zq$ and sends $g^{a}$ to the prover.

\item The prover sends $R_{P}=g^{[(g^{a})^{s_{i}}]^{-1}}$ to the verifier.

\item The verifier computes $R_{V}=g^{[(g^{s_{i}})^{a}]^{-1}}$ ($g^{s_{i}}$ $%
=$ $\prod_{j=0}^{k-1}(C_{j})^{i^{^{_{j}}}}$).

\item If $R_{V}=R_{P}$, the prover is the shareholder who possesses the
share $s_{i}$.
\end{enumerate}

\subsubsection{Pooling the shares}

The secret is reconstructed from the submitted shares, as follows: $%
s=\sum_{i=1}^{k}w_{i}s_{i}$ where $w_{i}=\sum_{i\neq j}i/(j-1)$.

Notice that the shares can be submitted using the same encryption function
of the distribution process $(E_{i}=s_{i}\oplus (g^{s_{i}})^{a})$ where $a$
is the private key of the party concerned by the reconstruction of the
secret and $g^{a}$ is its public key.

Notice also that this party does not need to run the membership process
before the pooling phase since using this encryption function allows the
verification of a share and its extraction at the same time.

\section{Security}

In this section, we prove the security properties of our PVSS scheme. First
of all, we provide our definition of a secure PVSS scheme:

\begin{definition}
A PVSS scheme is secure if and only if:

- During the dealing phase, neither the dealer $D$ can cheat by sending an
invalid share to a given participant $Pr_{i}$, nor the participant $Pr_{i}$
can claim that he received a non valid share while it was.

- During the reconstruction phase, an unauthorized party cannot pretend to
be a shareholder.

- During all the stages of the scheme, the secrecy property is verified.
\end{definition}

Let's prove at first that, in our scheme, the dealer $D$ cannot cheat by
sending an invalid share to the participant $Pr_{i}$. We show here that $%
Pr_{i}$ can prove this fact to the third party $R$ in the disputation phase.
Thus, we prove the following lemma:

\begin{lemma}
\textquotedblleft The dealer $D$ cannot cheat by sending an invalid share to
the participant $Pr_{i}$\textquotedblright .
\end{lemma}

\begin{proof}
In the disputation phase, a honest dealer has to compute $\lambda
=s_{i}\oplus (g^{a_{i}})^{s_{i}}$. But a malicious dealer can have another
behavior. In fact, he can compute $\lambda $ using an invalid share $%
s_{i}^{\prime }$ or an incorrect value $g^{a_{i}^{\prime }}$ rather than the
public key $g^{a_{i}}$ of the participant $Pr_{i}$.

So, there are seven values of $\lambda $ that $D$ can use: $\lambda
=s_{i}^{\prime }\oplus (g^{a_{i}})^{s_{i}}$ or $\lambda =s_{i}\oplus
(g^{a_{i}})^{s_{i}^{\prime }}$ or $\lambda =s_{i}^{\prime }\oplus
(g^{a_{i}})^{s_{i}^{\prime }}$ or $\lambda =s_{i}\oplus (g^{a_{i}^{\prime
}})^{s_{i}}$ or $\lambda =s_{i}^{\prime }\oplus (g^{a_{i}^{\prime
}})^{s_{i}} $ or $\lambda =s_{i}\oplus (g^{a_{i}^{\prime }})^{s_{i}^{\prime
}}$ or $\lambda =s_{i}^{\prime }\oplus (g^{a_{i}^{\prime }})^{s_{i}^{\prime
}}$.

In each of these cases, $Pr_{i}$ will compute $\alpha $ $=$\ $\lambda \oplus
(g^{s_{i}})^{a_{i}}$ at step $3$ of the disputation process, and since $%
\lambda \neq s_{i}\oplus (g^{a_{i}})^{s_{i}}$, he will find $\alpha $ $\neq $
$s_{i}$ and he will send this value to $R$.

$R$ will verify that $g^{\alpha }\neq
\prod_{j=0}^{k-1}(C_{j})^{i^{^{_{j}}}}$ and that $g^{1/(\lambda \oplus
\alpha )}=g^{1/(\lambda \oplus \lambda \oplus
(g^{s_{i}})^{a_{i}})}=g^{[(g^{a_{i}})^{^{s_{i}}}]^{-1}}$. So $R$ will
conclude that $D$ lied.
\end{proof}

We prove also that, in our scheme, a malicious behavior of a participant $%
Pr_{i}$, who received a valid share from the dealer $D$, but claims that his
share is invalid, will be detected. We show here that, in the disputation
phase, the dealer $D$ can prove to a third party $R$ that $Pr_{i}$ cheated. Thus, we prove the following lemma:

\begin{lemma}
\textquotedblleft The participant $Pr_{i}$, cannot claim that he received a
non valid share while it was\textquotedblright .
\end{lemma}

\begin{proof}
In the disputation phase, if a participant $Pr_{i}$ received a correct share $s_{i}$ but
claims that he received an invalid one, he has to send a fake value $\alpha
^{\prime }$ to $R$. In fact, $Pr_{i}$ computes $\alpha =\lambda \oplus
(g^{s_{i}})^{a_{i}}$ but sends $\alpha ^{\prime }$ $\neq \alpha $ to $R$.
So, $R$ computes, $g^{\alpha ^{\prime }}$ and verifies that it is not a
public value. Then, $R$ verifies, at step $5$ of the disputation process,
that $g^{1/(\lambda \oplus \alpha ^{\prime })}\neq g$ $%
^{[(g^{a_{i}})^{s_{i}}]^{-1}}$. Since it does hold, $R$ concludes that $%
Pr_{i}$ lied.
\end{proof}

In addition, we prove the following lemma:

\begin{lemma}
\textquotedblleft Under the Computational Diffie-Hellman assumption, it is
infeasible to break the encryption of the shares\textquotedblright .
\end{lemma}

\begin{proof}
Breaking the encryption of the shares is equivalent to computing $s_{i}$
from the encrypted share $E_{i}=s_{i}\oplus (g^{a_{i}})^{s_{i}}$. 

To be able to do that, we have to compute $s_{i}=E_{i}\oplus (g^{a_{i}})^{s_{i}}$ from
the inputs $E_{i}$, $g^{a_{i}}$, $g^{s_{i}}$. This implies computing $%
g^{a_{i}\ast s_{i}}$ given $g^{a_{i}}$ and $g^{s_{i}}$. 

Recall that the Computational Diffie-Hellman assumption states that it is infeasible to
compute $g^{a_{i}\ast s_{i}}$ given $g^{a_{i}}$ and $g^{s_{i}}$. Therefore
the unauthorized party is not able to compute the share $s_{i}$.

Furthermore, to break the encryption of a share $s_{i}$, the adversary
should be able to compute $s_{i}$ from $g^{s_{i}}$. This implies solving the
Discrete Logarithm Problem. 

Given that computing the discrete log in $G_{q}$
is infeasible, the unauthorized party is not able to compute $s_{i}$ from $%
g^{s_{i}}$.
\end{proof}

Then, we prove the following lemma:

\begin{lemma}
\textquotedblleft Under the Computational Diffie-Hellman assumption, an
unauthorized party cannot extract the share $s_{i}$ from $g^{a_{i}}$, $%
g^{s_{i}}$ and the published masked value $\lambda$ in the
disputation process\textquotedblright .
\end{lemma}

\begin{proof}
To extract the share $s_{i}$, the adversary has to compute $s_{i}$ from the
public masked value $\lambda =s_{i}\oplus g^{a_{i}\ast s_{i}}$. This implies
that he needs to compute $s_{i}=\lambda \oplus $ $g^{a_{i}\ast s_{i}}$ given 
$\lambda $, $g^{a_{i}}$ and $g^{s_{i}}$. 

For doing that, the adversary
should be able to compute $g^{a_{i}\ast s_{i}}$ from the inputs $%
g^{a_{i}}$ and $g^{s_{i}}$. However, the adversary is not able to compute $%
s_{i}$ due to the Computational Diffie-Hellman assumption.
\end{proof}

We prove also that:

\begin{lemma}
\textquotedblleft Under the Computational Diffie-Hellman assumption, an unauthorized
party cannot retrieve the share $s_{i}$ from $g^{a_{i}}$, $g^{s_{i}}$%
and $g^{[(g^{s_{i}})^{a_{i}}]^{-1}}$\textbf{\ }in the two first steps of the
disputation process\textquotedblright .
\end{lemma}

\begin{proof}
Under the assumption that computing Discrete Logarithm in $G_{q}$ is hard,
an unauthorized party cannot extract $g^{a_{i}\ast s_{i}}$ from $%
g^{[(g^{s_{i}})^{a_{i}}]^{-1}}$ and under the Computational Diffie-Hellman
assumption, it is not possible to retrieve $s_{i}$ from $g^{s_{i}}$and $%
g^{a_{i}}$.
\end{proof}

Moreover, we prove that:

\begin{lemma}
\textquotedblleft Under the Computational Diffie-Hellman assumption, an
unauthorized party cannot pretend to be a shareholder\textquotedblright .
\end{lemma}

\begin{proof}
This feature is fulfilled within the membership process. In this process, to
pretend to be the shareholder possessing $s_{i}$, the unauthorized party
should be able to compute $(g^{a_{i}})^{s_{i}}$ from the values $g^{a_{i}}$
and $g^{s_{i}}$ in the membership process. However, under the Computational
Diffie-Hellman assumption, this is infeasible.
\end{proof}

Finally, we prove that:

\begin{lemma}
\textquotedblleft Under the Computational Diffie-Hellman assumption, it is
infeasible to break the encryption of the shares submitted in the
reconstruction phase\textquotedblright .
\end{lemma}

\begin{proof}
In the reconstruction phase, only the party possessing the private key $a$\
can extract the share $s_{i}$\ from the encrypted value $E_{i}=s_{i}\oplus
(g^{a})^{s_{i}}$. This party has just to compute $s_{i}=E_{i}\oplus \lbrack
(g^{s_{i}})^{a}]$. 

For a dishonest party knowing only $E_{i}$, $g^{a}$\ and $%
g^{s_{i}}$, breaking the encryption of the shares means computing $%
(g^{a})^{s_{i}}$\ from the public value $g^{s_{i}}$\ and the public key $%
g^{a}$\ which is infeasible under the Computational Diffie-Hellman
assumption.
\end{proof}

In this section, we proved that neither the dealer can cheat by distributing
invalid shares nor a dishonest participant can cheat by claiming that the
share he received is not valid while it was. Moreover, we proved that under
the Computational Diffie-Hellman assumption, no one can break the encryption
of the shares neither in the distribution process, nor in the disputation
process or in the reconstruction phase. We proved also that, under the
Computational Diffie-Hellman assumption, an unauthorized party cannot
pretend to be a shareholder possessing a valid share.

In the following section, we compare our new PVSS scheme to the PVSS schemes
presented in section 2.

\section{Comparison with previous PVSS schemes}

In this section, in order to compare our PVSS scheme to the existent PVSS
schemes, we first present the different security properties of the most
known schemes. We point that the schemes proposed in [12] and [18] do not
appear in this section because we consider that these schemes have a
specific context\footnote{
process fast or delayed, and the scheme proposed in [21] focused on how to
make a new member join the scheme without exposing the secret and the old
shares.}. However, we include the scheme of Feldman [6] in this comparison
since we consider that it is the first PVSS scheme, although public
verifiability was not defined yet when this scheme was proposed. So, for each studied PVSS scheme, we identify the cryptographic techniques it uses in every process (distribution, verification\ldots ) and we verify
if they satisfy our definition of security. Since most of the used
cryptographic techniques are based on some hard problems, we classify these
hard problems into four classes:

\begin{itemize}
\item Discrete Logarithms: Hard problems based on the Discrete Logarithm
Problem.
\item Factoring : Hard problems based on the Factorization Problem.
\item Paillier's cryptosystem: Hard problems based on the Paillier's
cryptosystem proof.
\item Pairings: Hard problems based on the Bilinear Pairings.
\end{itemize}

As we said before, a comparison is done for every process of PVSS
schemes. For the distribution process, we study the security assumptions (DLP%
\footnote{%
Discrete Logarithm Problem.}, CDH\footnote{%
Computational Diffie-Hellman Problem.}, ...) of the encryption functions
used to encrypt the shares before distributing them among the set of
participants. Then, we evaluate the problem on which the security of the
process is based. The evaluation is based on the following reduction:
ELGamal $\leq _{P}$ CDH $\leq _{P}$ DLP.

However, for the scheme of Feldman and the scheme of Young and Yung, this
evaluation is infeasible, because the cryptographic techniques used in these
schemes are not specified. For more details, see table1.

\begin{table*}[h] \centering%
\begin{tabular}{|p{3cm}||p{4cm}||p{2.5cm}||p{2.5cm}|}
\hline
\multicolumn{4}{|c|}{\textbf{Encryption and distribution of shares}} \\ 
\hline
\textbf{Category} & \textbf{PVSS Scheme} & \textbf{Problem} & \textbf{%
Evaluation} \\ \hline
\multirow{7}{*}{Discrete Log} & Stadler (1996) & ELGamal cryptosystem & Hard \\ 
\cline{2-4}
& Schoenmakers (1999) & DLP & Very hard \\ \cline{2-4}
& Behnad \& Eghlidos (2008) & CDH & Hard \\ \cline{2-4}
& Heidarvand \&Villar (2009) & DLP & Very hard \\ \cline{2-4}
& Jhanwar (2010) & DLP & Very hard \\ \cline{2-4}
& Ben Shil, Blibech \& Robbana (2011) & CDH & Hard \\ \cline{2-4}
& Our PVSS (2012) & CDH & Hard \\ \hline
\multirow{1}{*}{Factoring} & Okamoto \& Fujisaki (1998) & Modified RSA assumption
& Non Proved \\ \hline
Paillier cryptosystem & Ruiz \& Villar (2005) & Paillier
probabilistic encryption scheme & Hard \\ \hline
\multirow{1}{*}{Pairings} & Wu \& Tseng (2011) & BDH & Hard \\ \hline
\multirow{2}{3cm}{Non specified problem} & Feldman (1987) & No encryption function & 
- \\ \cline{2-4}
& Young \& Yung (2001) & Public key encryption algorithm & - \\ \hline
\end{tabular}%
\caption{Evaluation of the distribution process}\label{Evaluation of the
distribution process}%
\end{table*}%
For the verification process, we explicit also the problem on which the
security of the verification process is based. This evaluation is based on
the following reductions:

\begin{itemize}
\item ELGamal $\leq _{P}$ CDH $\leq _{P}$ DLP.
\item RSA $\leq _{P}$Factoring.
\end{itemize}

We also classify the verification process into two classes: interactive
verification and non-interactive verification. The verification is
interactive if the verifier has to communicate with other participants
and/or with the dealer to verify the validity of a share. It is
non-interactive if the verifier can verify the validity of a share without
any communication with other participants or with the dealer. Obviously,
non-interactivity is preferred in order to reduce communications. For more
details, see table 4.

\begin{table*}[tbp] \centering%
\begin{tabular}{|p{2.5cm}||p{3cm}||p{1.5cm}||p{1.7cm}||p{1.6cm}||p{1.7cm}|}
\hline
\multicolumn{6}{|c|}{\textbf{Verification of shares}} \\ \hline
\textbf{Category} & \textbf{PVSS scheme} & \textbf{Problem} & \textbf{%
Evaluation} & \textbf{Proof} & \textbf{Evaluation} \\ \hline
\multirow{8}{2.5cm}{Discrete Log} & Feldman (1987) & DLP & Very hard & Non-interactive
& Standard Model \\ \cline{2-4}\cline{2-6}
& \multirow{2}{2.5cm}{Stadler (1996)} & \multirow{2}{*}{DDLP} & \multirow{2}{*}{Non proved} & 
Interactive & Zero-Knowledge \\ \cline{5-6}
&  &  &  & Non-interactive & Random Oracle Model \\ \cline{2-4}\cline{2-6}
& \multirow{2}{3cm}{Schoenmakers (1999)} & \multirow{2}{*}{DDH} & \multirow{2}{*}{Hard} & 
Interactive & Zero-Knowledge \\ \cline{5-6}
&  &  &  & Non-interactive & Random Oracle Model \\ \cline{2-4}\cline{2-6}
& \multirow{2}{3cm}{Young \& Yung (2001)} & \multirow{2}{*}{DLP} & \multirow{2}{*}{Very hard} & 
Interactive & Zero-Knowledge \\ \cline{5-6}
&  &  &  & Non-interactive & Random Oracle Model \\ \cline{2-4}\cline{2-6}
& Behnad \& Eghlidos (2008) & DLP & Very hard & Non-interactive & Standard
Model \\ \cline{2-6}
& Wu \& Tseng (2011) & CDH & Hard & Non-interactive & Random Oracle Model \\ 
\cline{2-6}
& Ben Shil, Blibech \& Robbana (2011) & DLP & Very hard & Non-interactive & 
Standard Model \\ \cline{2-6}
& Our PVSS (2012) & DLP & Very hard & Non-interactive & Standard Model \\ 
\hline
\multirow{2}{2.5cm}{Factoring} & \multirow{2}{3cm}{Okamoto \& Fujisaki (1998)} & Factoring & 
Very hard & Interactive & Zero-Knowledge \\ \cline{3-6}
&  & RSA & Hard & Interactive & Zero-Knowledge \\ \cline{1-4}\cline{2-6}
Paillier cryptosystem & Ruiz \& Villar (2005) & DCRA & Hard & 
Non-interactive & Random Oracle Model \\ \hline
\multirow{2}{2.5cm}{Pairings} & Heidarvand \& Villar (2009) & DBSDH & Hard & 
Non-interactive & Standard Model \\ \cline{2-6}
& Jhanwar (2010) & MSEDH & Hard & Non-interactive & Standard Model \\ \hline
\end{tabular}%
\caption{Evaluation of the verification process}\label{Evaluation of the
verification process}%
\end{table*}%

After the verification process, a participant can initiate a disputation
process to complain about the validity of the share he received. The
disputation process aims to verify if the dealer is honest. We say that this
process is explicit if it leads the dealer to send the share to the
participant who complains in the presence of a third party. This latter has
to identify who among the dealer and the participant is lying. Otherwise,
the disputation process is supposed to be implicit (the dealer is considered
as dishonest if the number of participants complaining about the validity of
their shares is greater than a given parameter). Notice that only the three
schemes of table 2 offer an explicit disputation process. The security
assumptions of this process for these schemes are studied in table 3.

\begin{table*}[tbp] \centering%
\begin{tabular}{|p{1.9cm}||p{3cm}||p{1.5cm}||p{1.7cm}||p{1.7cm}|}
\hline
\multicolumn{5}{|c|}{\textbf{Disputation}} \\ \hline
\textbf{Category} & \textbf{PVSS scheme} & \textbf{Problem} & \textbf{%
Evaluation} & \textbf{Proof} \\ \hline
\multirow{1}{*}{Discrete Log} & Behnad \& Eghlidos (2008) & CDH & Hard & 
Interactive \\ \cline{2-5}
& Ben Shil, Blibech \& Robbana (2011) & CDH & Hard & Interactive \\ 
\cline{2-5}
& Our PVSS (2012) & CDH & Hard & Interactive \\ \hline
\end{tabular}%
\caption{Evaluation of the disputation process}\label{Evaluation of the
disputation process}%
\end{table*}%
The membership proof can be implicit (a participant has to give his part, in
the reconstruction process, to prove that he is an authorized participant)
or explicit (a participant can prove to a verifier that he is an authorized
participant possessing a valid share without revealing this share). When
this process is explicit, it can be interactive or non-interactive. In table
3, we focus on PVSS schemes with explicit membership proof process and study
the interactivity of each process and its security assumptions.

\begin{table*}[tbp] \centering%
\begin{tabular}{|p{1.9cm}||p{3.2cm}||p{1.5cm}||p{1.7cm}||p{1.5cm}||p{1.7cm}|}
\hline
\multicolumn{6}{|c|}{\textbf{Membership proof}} \\ \hline
\textbf{Category} & \textbf{PVSS scheme} & \textbf{Problem} & \textbf{%
Evaluation} & \textbf{Proof} & \textbf{Evaluation} \\ \hline
\multirow{4}{*}{Discrete Log} & \multirow{2}{*}{Schoenmakers (1999)} & \multirow{2}{*}{DDH} & %
\multirow{2}{*}{Hard} & Interactive & Zero-Knowledge \\ \cline{5-6}
&  &  &  & Non-interactive & Random Oracle Model \\ \cline{2-6}\cline{5-6}
& Behnad \& Eghlidos (2008) & CDH & Hard & Interactive & Zero-Knowledge \\ 
\cline{2-6}
& Ben Shil, Blibech \& Robbana (2011) & CDH & Hard & Interactive & 
Zero-Knowledge \\ \cline{2-6}
& Our PVSS (2012) & CDH & Hard & Interactive & Zero-Knowledge \\ \hline
\multirow{1}{*}{Pairings} & Heidarvand \&Villar (2009) & DBSDH & Hard & 
Non-interactive & Standard Model \\ \hline
\end{tabular}%
\caption{Evaluation of the membership proof}\label{Evaluation of the
membership proof}%
\end{table*}%
To summarize, we provide in this paper a new PVSS scheme having the
following properties:

First, during the distribution process, our scheme uses a simple encryption
function to encrypt the shares before distributing them. The encryption of
the shares is secure under the CDH assumption.

When he receives a share of the secret, a participant can extract and verify
the validity of his share without any communication with any party, even the
dealer. We say that our verification process is non-interactive.

In case of any complaint against the dealer, the concerned participant, the
dealer and a third party $R$ can run a disputation process in order to
establish who is cheating. The disputation process is secure under the CDH
assumption.

Later, an explicit Zero-Knowledge membership process can be run to allow
every participant to prove interactively his membership to a verifier who
asked for that.\textbf{\ }This process is secure under the CDH assumption. Notice here that only three schemes offer an explicit membership proof and an explicit disputation process at the same time: the present scheme, the
scheme of Behnad and Eghlidos [2] and the scheme of Ben Shil, Blibech and
Robbana [17].

Moreover, notice that in our scheme, when submitting an encrypted share to
the party concerned by computing the secret, an implicit membership proof is
given and it is not necessary to run the explicit membership only proof.

Finally, we point that the use of the XOR operator in our scheme makes it
less timeconsuming than the schemes presented in [2] and [17].

\section{Conclusion}

The new PVSS scheme proposed in this paper is very simple while being
secure. In fact, thanks to the use of a simple encryption function, we
reduce computations in all the processes of the scheme. In addition, like in
the scheme proposed in [2] we added two new processes: a disputation
process and a membership proof process. Thanks to these processes, no one
can cheat.

\nocite{*}
\bibliographystyle{eptcs}
\bibliography{generic}
\end{document}